\documentclass[letterpaper,11pt]{article}

\usepackage{bm}
\usepackage{url}
\usepackage{array}
\usepackage{color}
\usepackage{amsthm}
\usepackage{dsfont}
\usepackage{amsmath}
\usepackage{amssymb}
\usepackage{authblk}
\usepackage{caption}
\usepackage{amsfonts}
\usepackage{booktabs}
\usepackage{graphicx}
\usepackage{mathrsfs}
\usepackage{multicol}
\usepackage{multirow}
\usepackage{cellspace}
\usepackage{enumerate}
\usepackage{enumitem}
\usepackage{makecell}

\usepackage{hyperref}
\hypersetup{
	linktoc      = all,
	colorlinks   = true,
	urlcolor     = blue,
	linkcolor    = blue,
	citecolor    = red
}

\usepackage[backend=biber, isbn=false, style=alphabetic, backref=true, doi=true, url=false, maxcitenames=10, mincitenames=10, maxalphanames=10, maxbibnames=10, minbibnames=5, minalphanames=5, defernumbers=true, sortlocale=en_US]{biblatex}
\usepackage[T1]{fontenc}
\usepackage[dvipsnames]{xcolor}
\usepackage[margin=1in]{geometry}
\usepackage[capitalize,noabbrev]{cleveref}
\usepackage[ruled,vlined,linesnumbered]{algorithm2e}
\usepackage[caption=false,font=normalsize,labelfont=sf,textfont=sf]{subfig}

\usepackage{titlesec}
\titlespacing*{\paragraph}{0pt}{1ex}{1ex}

\setlist[itemize]{leftmargin=*}
\setlist[enumerate]{leftmargin=*}
\setlist[enumerate]{label=(\arabic*)}

\newtheorem{theorem}{Theorem}[section]
\newtheorem{lemma}[theorem]{Lemma}

\newtheorem{definition}[theorem]{Definition}

\crefname{algocf}{Algorithm}{Algorithms}
\Crefname{algocf}{Algorithm}{Algorithms}

\crefname{algocfline}{Line}{Lines}
\Crefname{algocfline}{Line}{Lines}

\crefname{lemma}{Lemma}{Lemmas}
\Crefname{lemma}{Lemma}{Lemmas}

\DeclareMathOperator{\polylog}{polylog}
\DeclareMathOperator{\nnz}{nnz}
\DeclareMathOperator{\vol}{vol}
\DeclareMathOperator{\supp}{supp}

\DeclareMathOperator{\tvol}{\widetilde{vol}}

\def\eps{\varepsilon}

\def\R{\mathbb{R}}

\def\tO{\widetilde{O}}

\def\mat{\mathbf}

\def\A{\mat{A}}
\def\D{\mat{D}}

\def\L{\mat{L}}
\def\M{\mat{M}}
\def\N{\mat{N}}
\def\S{\mat{S}}
\def\mzero{\mat{0}}

\def\vec{\boldsymbol}

\def\b{\vec{b}}
\def\d{\vec{d}}
\def\e{\vec{e}}
\def\p{\vec{p}}
\def\r{\vec{r}}
\def\s{\vec{s}}
\def\x{\vec{x}}
\def\y{\vec{y}}

\def\vdelta{\vec{\delta}}

\def\vzero{\vec{0}}
\def\vone{\vec{1}}

\def\pr{\vec{\mathrm{pr}}}

\def\tp{\tilde{\p}}
\def\tx{\tilde{\x}}
\def\tr{\tilde{\r}}

\def\xast{\x^{\ast}}
\def\Sast{S^{\ast}}

\def\PageRankNibble{\textup{\texttt{PageRank-Nibble}}\xspace}
\def\push{\textup{\texttt{push}}\xspace}
\def\activeSet{\textup{\texttt{ActiveSetPageRank}}\xspace}

\title{A Simple Active-Set Method for PageRank-Based \\ Local Graph Clustering}

\author{Zhewei Wei}
\author{Mingji Yang}

\affil{Renmin University of China \authorcr
	\{zhewei,kyleyoung\}@ruc.edu.cn}

\date{}

\begin{document}

\maketitle

\begin{abstract}

Local graph clustering aims to find a well-connected cluster near a given seed node without exploring the entire graph.
A key step in the classic local clustering algorithm of Andersen, Chung, and Lang (ACL; Internet Math. 2007) is to approximate the PageRank vector from the seed node.
Their local push method computes an ACL $\eps$-approximate PageRank vector with teleportation parameter $\alpha$ in $O\bigl(1/(\alpha\eps)\bigr)$ time.

We give an algorithm that computes an ACL $\eps$-approximate PageRank vector in $\tO\bigl(1 / \eps^2\bigr)$ time with high probability.
This bound is independent of the graph size and has only a polylogarithmic dependence on $1 / \alpha$, albeit with a quadratic dependence on $1 / \eps$.
As a direct consequence, we obtain a new running-time tradeoff between the target conductance and target volume in local graph clustering.
Our method also applies to the optimization problem of $\ell_1$-regularized PageRank and computes an additive approximate minimizer with a polylogarithmic dependence on $1/\alpha$, improving the $1/\sqrt{\alpha}$ dependence in the previous bound of Mart\'inez-Rubio, Wirth, and Pokutta (COLT 2023).

Our algorithm is based on an intuitive process that maintains a growing active set of nodes: it performs push operations on the current set until convergence and then expands the set and repeats the process if necessary.
We show that for each active set, the corresponding limiting state is the solution to a symmetric diagonally dominant (SDD) linear system on the set.
We apply nearly-linear-time SDD solvers to these systems and prove that the approximation preserves the properties of the push process.

\end{abstract}

\section{Introduction} \label{sec:intro}

Given a seed node in an undirected graph, local graph clustering aims to find a well-connected cluster near the node without exploring the entire graph.
In the influential work by Spielman and Teng~\cite{spielman2004nearly}, they initiated the study of local clustering algorithms and applied them to graph partitioning, graph spectral sparsification, and nearly-linear-time Laplacian solvers.

Subsequently, Andersen, Chung, and Lang~\cite{andersen2006local,andersen2007pagerank} introduced the improved method \PageRankNibble, which computes an approximate PageRank vector~\cite{brin1998anatomy} from the seed node and uses it to identify a cluster with small conductance.
The analysis for the clustering guarantees was later simplified in \cite{andersen2007detecting}.
In this paper, we follow this PageRank-based local clustering framework and focus on the complexity of computing an approximate PageRank vector that can yield local clustering guarantees.

For a seed node $s$ and teleportation parameter $\alpha \in (0,1)$, the PageRank value from $s$ to a node $v$ can be viewed as the probability that a random walk that starts from $s$ terminates at $v$, where it terminates with probability $\alpha$ at each step.
The framework of \cite{andersen2007pagerank,andersen2007detecting} uses a specific notion of $\eps$-approximate PageRank vector, which is later defined in \cref{def:ACL}.
The classic \push algorithm in \cite{andersen2007pagerank} computes an $\eps$-approximate PageRank vector in $O\bigl(1 / (\alpha\eps)\bigr)$ time, but its query complexity is only $O(1 / \eps)$.
This raises the question of whether the linear dependence on $1 / \alpha$ in their running-time bound can be improved.

In this paper, we present a simple active-set method that computes an $\eps$-approximate PageRank vector in $\tO(1 / \eps^2)$ time with high probability, showing that the complexity can have only a polylogarithmic dependence on $1 / \alpha$ while remaining independent of the graph size.
This gives a new running-time tradeoff between $\alpha$ and $\eps$ for PageRank computation, and furthermore, applying this result to the local clustering algorithm of \cite{andersen2006local,andersen2007detecting} directly yields a new tradeoff between the target conductance and target volume.

Recently, PageRank-based local clustering has also gained attention from a variational perspective.
\cite{fountoulakis2019variational} proposed an $\ell_1$-regularized variational formulation of approximate PageRank computation and showed that a standard optimization method solves the problem locally with a $1 / \alpha$ dependence in the running time.
Later, Fountoulakis and Yang~\cite{fountoulakis2022open} asked whether accelerated optimization methods could improve the dependence of $1 / \alpha$ to $1 / \sqrt{\alpha}$ without destroying locality.
Mart\'inez-Rubio, Wirth, and Pokutta~\cite{rubio2023accelerated} showed that achieving the $1 / \sqrt{\alpha}$ dependence is possible if we allow a larger dependence on the support size of the optimal solution.
By applying our method to $\ell_1$-regularized PageRank, we show that if such a larger dependence on the support size is allowed, then the $1 / \sqrt{\alpha}$ factor can be improved to $\polylog(1/\alpha)$.

\subsection{Problem Formulation}

Let $G = (V,E)$ be a connected, unweighted, undirected simple graph with $n := |V| \ge 2$ nodes.
Let $\A$ be its adjacency matrix, $\D$ be its diagonal degree matrix, and $\d$ be its degree vector.
For a node $v \in V$, let $N(v)$ denote the set of its neighbors, and let $\e_v \in \R^V$ denote its standard basis vector.
We define the volume $\vol(S) := \sum_{v \in S}\d(v)$ and internal volume $\tvol(S) := |S| + \bigl|\bigl\{(u,v) \in E: u \in S, v \in S\bigr\}\bigr|$ for $S \subseteq V$, and let $\supp(\x) := \bigl\{v : \x(v) \ne 0\bigr\}$ denote the support for $\x \in \R^V$.
For $\varnothing \subsetneq S \subsetneq V$, we define the conductance
\begin{align*}
	\phi(S) := \frac{\bigl|\bigl\{(u,v) \in E : u \in S, v \in V \setminus S\bigr\}\bigr|}{\min\bigl\{\vol(S),\vol(V\setminus S)\bigr\}}.
\end{align*}
For the local algorithms, we assume the standard adjacency-list model, where each degree query takes $O(1)$ time and each neighbor query returns one neighbor in $O(1)$ time.

For any source vector $\s \in \R^V$ and a \textit{teleportation parameter} $\alpha \in (0,1)$, we use $\pr_{\alpha}(\s) \in \R^V$ for its PageRank vector, which is defined as the unique solution to the linear system
\begin{align}
	\pr_{\alpha}(\s) = \alpha\s + (1-\alpha) \A\D^{-1} \pr_{\alpha}(\s). \label{eqn:PageRank}
\end{align}
Note that we do not regard $\alpha$ as a constant in this paper.
We focus on the following sense of approximate PageRank vector as considered by \cite{andersen2007pagerank,andersen2007detecting}.

\begin{definition}[ACL $\eps$-approximate PageRank vector~\cite{andersen2007pagerank,andersen2007detecting}] \label[definition]{def:ACL}
	A vector $\p \in \R^V_{\ge 0}$ is called an $\eps$-approximate PageRank vector (w.r.t. seed node $s$ and teleportation parameter $\alpha$) if there exists a vector $\r$ such that $\p = \pr_{\alpha}(\e_s-\r)$ and $\vzero \le \r \le \eps \d$ entrywise.
\end{definition}

Our main goal is to compute an ACL $\eps$-approximate PageRank vector with support volume $O(1 / \eps)$.
We also study the following $\ell_1$-regularized PageRank problem, which is equivalent to the ones in previous works~\cite{fountoulakis2019variational,fountoulakis2022open,rubio2023accelerated} up to a variable scaling.

\paragraph{$\ell_1$-regularized PageRank.}
Given $G$, $s$, $\alpha$, and a regularization parameter $\rho > 0$, the $\ell_1$-regularized PageRank problem~\cite{fountoulakis2019variational} is to minimize
\begin{align}
	\psi(\x) := \frac{1}{2} \, \x^{\top} \bigl(\D - (1 - \alpha)\A\bigr) \x - \alpha\,\e_s^{\top}\x + \alpha\rho\|\D\x\|_1 \label{eqn:regularized-objective}
\end{align}
over $\x \in \R^V$.
It is shown in \cite{fountoulakis2019variational} that $\psi(\cdot)$ has a unique minimizer $\xast \ge \vzero$ and $\D\xast$ is an ACL $\rho$-approximate PageRank vector.
We study the problem of computing a $\xi$-additive approximate minimizer of $\psi(\cdot)$ for a given $\xi > 0$, i.e., computing a vector $\tx \in \R^V$ such that $\psi(\tx) - \psi(\xast) \le \xi$.

\subsection{Our Results}

Our main result is a randomized algorithm that computes an ACL $\eps$-approximate PageRank vector in $\tO(1 / \eps^2)$ time with high probability, stated as follows.

\begin{theorem} \label{thm:ACL}
	Given a seed node $s \in V$, $0 < \alpha < 1$, $\eps > 0$, and $0 < \delta < 1$, there exists an algorithm that satisfies the following properties with probability at least $1 - \delta$:
	\begin{itemize}
		\item it returns an ACL $\eps$-approximate PageRank vector with support volume at most $2/\eps$;
		\item its running time is
		\begin{align*}
			O\Bigl( \frac{1}{\eps^2} \polylog\frac{1}{\alpha\eps\delta} \Bigr) = \tO\Bigl( \frac{1}{\eps^2} \Bigr).
		\end{align*}
	\end{itemize}
\end{theorem}

In comparison, the original \push algorithm~\cite{andersen2007pagerank} computes an ACL $\eps$-approximate PageRank vector in deterministic $O\bigl(1/(\alpha\eps)\bigr)$ time with support volume $O(1/\eps)$.
Therefore, our result improves the dependence on $1/\alpha$ from linear to polylogarithmic, but at the cost of a quadratic dependence on $1/\eps$.
This is not a uniform improvement, and the two algorithms emphasize different tradeoffs between $\alpha$ and $\eps$.
Ignoring polylogarithmic factors, our bound is better than $O\bigl(1/(\alpha\eps)\bigr)$ when $\alpha = o(\eps)$.
Also, note that the $\tO(1/\eps^2)$ time bound is still independent of the graph size.

\paragraph{Consequence for local graph clustering.}
The local clustering approaches in \cite{andersen2007pagerank,andersen2007detecting} use as a black box an algorithm that computes an ACL $\eps$-approximate PageRank vector with support volume $O(1 / \eps)$.
Since \cref{thm:ACL} provides exactly this guarantee, we can replace the original \push algorithm by our algorithm in the local clustering framework, and this only affects the running time of the local clustering algorithm.
Specifically, we obtain that, for many seed nodes within a set $C$ of conductance $O\bigl(\phi^2/ \log^2 n \bigr)$, with probability at least $1-\delta$ we can find a cut with conductance at most $\phi$ and volume $O(k)$ in time
\begin{align*}
	O\Bigl(k^2 \log^2 n \polylog\frac{k \log n}{\phi \delta} \Bigr),
\end{align*}
provided that we can guess the volume $k$ of the smaller side of the cut within a factor of $2$.
For a formal statement, we refer the reader to \cite[Theorems~6.1 and 6.2]{andersen2007pagerank} or \cite[Theorems~2 and 3]{andersen2007detecting}.
In comparison, their complexities in the same statements are both $O\bigl(k \log^2 n / \phi^2\bigr)$.
Similarly, regarding the target conductance $\phi$ and target volume $k$, our result reduces the dependence on $1/\phi$ from quadratic to polylogarithmic, but increases the dependence on $k$ from linear to quadratic.

Our second result is an algorithm for computing an additive approximate minimizer of the $\ell_1$-regularized PageRank problem in \eqref{eqn:regularized-objective}.

\begin{theorem} \label{thm:regularized}
	Given a seed node $s \in V$, $\alpha \in (0,1)$, $\rho > 0$, $\xi > 0$, and $\delta \in (0,1)$ and letting $\xast$ be the unique minimizer for $\psi(\cdot)$ and $\Sast := \supp(\xast)$, there exists an algorithm that satisfies the following properties with probability at least $1 - \delta$:
	\begin{itemize}
		\item it returns $\tx$ such that $\psi(\tx)-\psi(\xast) \le \xi$, $\D\tx$ is an ACL $(2\rho)$-approximate PageRank vector, and $\vol\bigl(\supp(\tx)\bigr) \le 1 / \rho$;
		\item its running time is
		\begin{align*}
			O\Bigl( |\Sast|\tvol(\Sast)\polylog\frac{1}{\alpha\rho\xi\delta} + |\Sast|\vol(\Sast) \Bigr) = \tO\bigl(|\Sast|\vol(\Sast)\bigr) = \tO\Bigl(\frac{1}{\rho^2}\Bigr).
		\end{align*}
	\end{itemize}
\end{theorem}

For comparison, the ASPR algorithm in \cite{rubio2023accelerated} outputs a $\xi$-additive approximate minimizer in time
\begin{align*}
	O\Bigl( \frac{1}{\sqrt{\alpha}} \, |\Sast| \tvol(\Sast) \log\frac{1}{\alpha\rho\xi} + |\Sast| \vol(\Sast) \Bigr).
\end{align*}
Thus, up to logarithmic factors, our bound removes the $1 / \sqrt{\alpha}$ factor from the term involving $|\Sast| \tvol(\Sast)$.
Additionally, there is an important difference between the output guarantees.
Their accelerated algorithm is guaranteed to return an additive approximate minimizer, but this guarantee alone cannot be directly substituted into the existing ACL local clustering analyses.
In contrast, our algorithm returns $\tx$ such that $\D\tx$ is an ACL $(2\rho)$-approximate PageRank vector with support volume at most $1 / \rho$, so it can be used directly in the PageRank-based local clustering framework discussed above.

\cite{rubio2023accelerated} also gives a CDPR algorithm that computes the exact minimizer $\xast$ in time
\begin{align*}
	O\bigl(|\Sast|^3 + |\Sast|\vol(\Sast)\bigr)
\end{align*}
with space complexity $O\bigl(|\Sast|^2\bigr)$.
Since $\tvol(\Sast) = O\bigl(|\Sast|^2\bigr)$, our bound is no larger than this bound up to logarithmic factors.
On the other hand, their algorithm computes the exact minimizer, whereas our algorithm computes an additive approximate minimizer, so these two bounds provide different accuracy guarantees.

\paragraph{Remark (Using deterministic solvers).}
We remark that the randomness of our algorithm only comes from invoking randomized SDD solvers (see \cref{sec:overview}).
If we instead use deterministic almost-linear-time SDD solvers~\cite{chuzhoy2022deterministic}, then our whole algorithm becomes deterministic, and this removes the dependence on the failure probability $\delta$ in the running-time bounds in \cref{thm:ACL,thm:regularized} at the cost of an additional factor of $(1 / \eps)^{o(1)}$ and $|\Sast|^{o(1)}$, respectively.

\paragraph{Remark (Lazy and non-lazy random walks).}
Some related work defines PageRank with teleportation parameter $\alpha_{\mathrm{L}}$ using lazy random walks instead of ordinary random walks, e.g., \cite{andersen2007pagerank,fountoulakis2019variational,fountoulakis2022open,rubio2023accelerated}.
It is known that the lazy and non-lazy definitions give the same PageRank vector when $\alpha_{\mathrm L} = \alpha / (2 - \alpha)$~\cite{andersen2007pagerank}.
Consequently, all our results convert directly to the lazy-walk ones through a conversion between $\alpha_{\mathrm L}$ and $\alpha$.
Since $1 / \alpha_{\mathrm L} = (2 - \alpha) / \alpha = \Theta(1 / \alpha)$, the asymptotic dependence on the inverse teleportation parameter is unchanged.

\section{Overview} \label{sec:overview}

Our method can be intuitively understood based on a variant of the classic \push algorithm~\cite{andersen2007pagerank}.
This section provides a walk-through of the main ideas of our algorithm for computing an ACL $\eps$-approximate PageRank vector.
As preliminaries, we first introduce a combinatorial understanding of the original \push algorithm as a mass-propagation process.
After that, we describe our conceptual infinite-push process and outline how to turn it into an implementable algorithm.

\paragraph{$\alpha$-discounted random walk.}
Given the teleportation parameter $\alpha \in (0,1)$, the $\alpha$-discounted random walk is defined as the following Markov-chain process on $V$: at each step, the walk terminates at the current node with probability $\alpha$, and otherwise moves to a uniformly random neighbor of the current node.
It is known that for each $v \in V$, $\pr_{\alpha}(\e_s)(v)$ equals the probability that an $\alpha$-discounted random walk starting from $s$ terminates at $v$~\cite{avrachenkov2007monte}.

\paragraph{The \push operation.}
The original \push algorithm works by repeatedly performing \push operations.
A single \push operation can be understood as a mass-propagation step that corresponds to a step in the $\alpha$-discounted random walk.
Specifically, the algorithm maintains two vectors $\p,\r \in \R^V$, where the \textit{reserve vector} $\p$ records the probability mass that has been settled at each node, and the \textit{residue vector} $\r$ records the mass that has not been settled.
A \push operation on a node $v$ transfers $\alpha$ fraction of its residue $\r(v)$ to its reserve $\p(v)$ (corresponding to terminating at the current node with probability $\alpha$) and distributes the remaining $ (1-\alpha)$ fraction evenly among its neighbors (corresponding to walking to a random neighbor with probability $1 - \alpha$).
For later use, we slightly generalize the \push operation by allowing any amount $\Delta \in \bigl[0,\r(v)\bigr]$ of residue mass to be processed and give the pseudocode in \cref{alg:push}.

\begin{algorithm}[ht]
	\DontPrintSemicolon
	\caption{$\push(v,\alpha,\Delta)$~\cite{andersen2007pagerank}} \label{alg:push}
	\KwIn{$v \in V$, teleportation parameter $\alpha \in (0,1)$, amount $\Delta \in \bigl[0,\r(v)\bigr]$}
	$\r(v) \gets \r(v) - \Delta$ \;
	$\p(v) \gets \p(v) + \alpha\Delta$ \;
	\For{\textbf{\textup{each}} \textup{neighbor} $u$ \textup{of} $v$}
	{
		$\r(u) \gets \r(u) + (1 - \alpha)\Delta / \d(v)$ \;
	}
\end{algorithm}

\paragraph{The classic \push algorithm.}
The \push algorithm~\cite{andersen2007pagerank} repeatedly performs \push operations until the residue vector $\r$ satisfies $\vzero \le \r \le \eps \d$ entrywise.
Specifically, it starts from $\p := \vzero$ and $\r := \e_s$ and performs $\push\bigl(v,\alpha,\r(v)\bigr)$ whenever $\r(v) > \eps \, \d(v)$ for some $v$.
It is shown that the \push operations preserve the invariant $\p = \pr_{\alpha}(\e_s-\r)$ as required in the definition of an ACL approximate PageRank vector.
Clearly, throughout the \push algorithm, $\p$ is nondecreasing, $\r \ge \vzero$, and the total probability mass satisfies $\vone^{\top} \p + \vone^{\top} \r = 1$.
By the termination condition, it follows that the final $\r$ satisfies $\vzero \le \r \le \eps \d$ and thus the final $\p$ is an ACL $\eps$-approximate PageRank vector.

\paragraph{The conceptual infinite-push process.}
We consider a different conceptual strategy of performing the \push operations, which allows a better control of the nodes that have been pushed.
We maintain an active set $S$ of nodes, which is initialized to $\{s\}$, and perform an infinite number of \push operations only on nodes in $S$.
In the limiting state of this infinite-push process, all the residues on $S$ are small, and only nodes in the outer boundary $\partial S := \{v \in V\setminus S : N(v)\cap S \ne \varnothing\}$ can have nonzero residues that violate the constraint in the ACL condition.
Thus, we examine the residues on $\partial S$, and if none of them violate the constraint, we return the current $\p$ as an ACL $\eps$-approximation.
Otherwise, we expand $S$ to include those nodes that have large residues and repeat the infinite-push process on the new active set.
More specifically, the conceptual infinite-push algorithm consists of the following steps.
\begin{enumerate}
	\item Initialize $\p \gets \vzero$, $\r \gets \e_s$, and $S \gets \{s\}$.
	\item Loop over each node $v \in S$, and if $v$ has $\r(v) > 0.5\eps \, \d(v)$, perform $\push\bigl(v,\alpha,\r(v) - 0.5\eps \, \d(v)\bigr)$.
	Repeat this loop to the limit, so every node of $S$ has residue exactly $0.5\eps \, \d(v)$.
	\item Inspect the boundary $\partial S$ and set $T \gets \bigl\{v \in \partial S : \r(v) > \eps \, \d(v)\bigr\}$.
	\item If $T$ is empty, return $\p$ as an ACL $\eps$-approximation.
	Otherwise, set $S \gets S \cup T$ and jump to step~(2).
\end{enumerate}
We call one execution of steps (2)-(4) an \textit{iteration} of the algorithm.
Intuitively, this infinite-push process also ensures that $\p$ is nondecreasing, $\r$ is nonnegative, and $\vone^{\top} \p + \vone^{\top} \r = 1$.
Note that in step (2), we retain $0.5\eps \, \d(v)$ residue for each $v \in S$.
This strategy controls the volume of the active set $S$: at the end of any iteration, the total residue mass on $S$ equals $0.5\eps\vol(S)$, which is bounded above by $1$ since the total mass is $1$, so we have $\vol(S) \le 2 / \eps$.
As $|S| \le \vol(S)$ and each expansion in the algorithm adds at least one node to $S$, this strategy also guarantees that the number of iterations is at most $2 / \eps$.

The key observation in realizing this conceptual process as a real algorithm is that, in the limiting state, the invariant property relates $\p$ and $\r$ via an SDD linear system, allowing us to solve for $\p$ based on the value of $\r$ on $S$.
Thus, the infinite-push step above can be equivalently replaced by exactly solving a symmetric diagonally dominant (SDD) linear system on $S$.
To improve the efficiency, we solve the SDD systems approximately using nearly-linear-time SDD solvers~\cite{spielman2004nearly,spielman2014nearly}.
The time complexity of each SDD solve is nearly linear in the number of nonzero entries of the SDD system, which is $O\bigl(\vol(S)\bigr) = O(1 / \eps)$ in our case.
As the number of iterations is also $O(1 / \eps)$, we obtain a total running time of $\tO(1 / \eps^2)$.

In our analysis, we give algebraic proofs to show that the behavior of our algorithm matches the intuitive properties of the infinite-push process, and also show that the approximation error of the SDD solvers does not affect the correctness of the algorithm.

For the $\ell_1$-regularized PageRank problem, we use the same active-set algorithm with different parameter settings, and we give the error bound for the objective function as well as the complexity bound in terms of the support of the exact minimizer of the problem.
A key step is to show that the active set is always contained in $\Sast$, the support of the exact minimizer.
Thus, the number of iterations is at most $|\Sast|$ and the involved SDD systems have $O\bigl(\vol(\Sast)\bigr)$ nonzero entries.
This yields the desired complexity bound of $\tO\bigl(|\Sast|\vol(\Sast)\bigr)$.

Our active-set method is partly inspired by the optimization method in \cite{rubio2023accelerated} for $\ell_1$-regularized PageRank, which also uses an active-set strategy and solves subproblems with nonnegative constraints on each active set.
Compared to their method, our method has a combinatorial interpretation, uses significantly fewer optimization tools, and does not consider subproblems with nonnegative constraints.

\section{Other Related Work} \label{sec:related_work}

\paragraph{Evolving-set methods.}
Andersen, Oveis Gharan, Peres, and Trevisan~\cite{andersen2016almost} developed a different local clustering approach based on simulating a volume-biased evolving-set process.
For any $\eta > 0$ and any node set $A$ that has conductance at most $\phi$, for at least half of the seed nodes in $A$, their algorithm outputs a set of conductance $O\bigl(\sqrt{\phi / \eta}\bigr)$ with constant probability.
The expected ratio between the work of their algorithm and the volume of its output is $\vol(A)^{\eta} \, \phi^{-1/2} \polylog n$.
This improves the dependence on $\phi$ in the work-volume ratio over \PageRankNibble from $\phi^{-1}$ to $\phi^{-1/2}$, and achieves an almost Cheeger-optimal conductance guarantee.
In contrast, our result gives a different tradeoff within the PageRank-based framework and reduces the dependence on $\phi^{-1}$ to polylogarithmic.

\paragraph{PageRank computation.}
PageRank was originally proposed by Google as a ranking method for the web~\cite{brin1998anatomy} and has since been used broadly in network analysis and beyond~\cite{gleich2015pagerank}.
A large literature studies local estimation of PageRank values and variants; see the survey of \cite{yang2024efficient} and recent results such as \cite{wang2024revisiting,wang2024revisitinga,thorup2026pagerank}.
Specifically, \cite{wei2024approximating} studies the related problem of computing PageRank from a source node under absolute and degree-normalized absolute error guarantees.
However, their running times have a linear dependence on $1 / \alpha$, and also their error guarantee does not meet the requirement in the ACL sense, so it cannot be substituted directly into the local-clustering theorem.

\paragraph{SDD solvers.}
Nearly-linear-time solvers for SDD systems originate in the work of Spielman and Teng~\cite{spielman2004nearly} and have since then been refined in a long line of work.
Subsequent results have substantially simplified the machinery, including combinatorial solvers based on cycle updates~\cite{kelner2013simple} and simple solvers based on sparse approximate Gaussian elimination~\cite{kyng2016approximate}.
We use the solvers only as a black box and do not rely on a particular construction.
There are also works that study sublinear-time solvers for SDD systems and beyond~\cite{andoni2019solving,kwok2026solving}, but their running times depend polynomially on an appropriate condition parameter, which would introduce a polynomial dependence on $1 / \alpha$ for the systems encountered in our problems.

Besides, a recent work~\cite{fountoulakis2026complexity} studies the complexity of classic accelerated proximal-gradient methods for $\ell_1$-regularized PageRank.
It proves some hardness results for such algorithms and also new upper bounds under certain assumptions on the graph structure.
The scope of our work is different from theirs in that we consider using general algorithms and do not make additional assumptions on the graph structure.

\section{The Active-Set Algorithm} \label{sec:active-set}

In this section, we present our active-set algorithm for computing approximate PageRank vectors and prove \cref{thm:ACL}.
This algorithm will also be applied for the $\ell_1$-regularized PageRank problem in \cref{sec:regularized}.

\paragraph{Notation.}
Throughout the paper, for any $\x \in \R^V$ and $S \subseteq V$, we use both $\x_S \in \R^S$ and $\x|_S \in \R^S$ to denote the restriction of $\x$ to the coordinates in $S$; when we use $\le$ and $\ge$ between two vectors, we mean entrywise comparison.
For $S \subseteq V$, let matrices $\A_S$ and $\D_S$ be the principal submatrices of $\A$ and $\D$ indexed by $S$, respectively; also define $\L_{\alpha} := \D - (1 - \alpha) \A$ and $\L_{\alpha,S} := \D_S - (1 - \alpha) \A_S$.
Note that the diagonal entries of $\L_{\alpha,S}$ are the node degrees in the original graph.
We call a square matrix \textit{symmetric diagonally dominant} (SDD) if it is symmetric and each of its diagonal entry is at least the sum of the absolute values of the other entries in the same row.
It follows that $\L_{\alpha,S}$ is SDD for any $S \subseteq V$.

\subsection{Algorithm Description}

To design an implementable algorithm, we need to replace the conceptual infinite-push process in \cref{sec:overview} with a finite procedure.
In the limiting state, every node in the active set $S$ has residue exactly $0.5\eps \, \d(v)$, so this boils down to expressing the reserve vector associated with a fixed residue vector.
To this end, we derive the relationship between a degree-normalized reserve vector $\x = \D^{-1}\p$ and its associated residue vector $\r(\x)$ that satisfies the invariant $\D\x = \pr_{\alpha}\bigl(\e_s-\r(\x)\bigr)$.
By the definition of $\pr_{\alpha}(\cdot)$ in \cref{eqn:PageRank}, $\D\x = \pr_{\alpha}\bigl(\e_s-\r(\x)\bigr)$ is equivalent to
\begin{align*}
	\D\x = \alpha(\e_s - \r(\x)) + (1 - \alpha) \A\x,
\end{align*}
which is also equivalent to
\begin{align}
	\L_{\alpha} \, \x = \alpha \bigl(\e_s - \r(\x)\bigr) \label{eqn:residue-invariant}
\end{align}
since we write $\L_{\alpha} := \D - (1 - \alpha) \A$.
In light of this, for any $\x \in \R^V$, we define
\begin{align}
	\r(\x) := \e_s - \frac{1}{\alpha} \, \L_{\alpha} \, \x \label{eqn:residue-def}
\end{align}
so that \cref{eqn:residue-invariant} holds by definition.
Now, given a fixed residue vector $\r(\x)$, the degree-normalized reserve vector $\x$ can be recovered by solving the linear system \eqref{eqn:residue-invariant}.
Furthermore, if we are given $\r \in \R^S$ and want to compute a vector $\x \in \R^V$ such that $\r(\x)_S = \r$, we can solve the restricted linear system $\L_{\alpha,S} \, \x_S = \alpha \bigl(\e_s|_S - \r\bigr)$ and set $\x_{V\setminus S} := \vzero$.
Here, $\L_{\alpha,S}$ is SDD and invertible.

The definition of $\r(\x)$ also gives the following mass identity property.
Since $\vone^{\top} \A = \vone^{\top} \D$, we have $\vone^{\top} \L_{\alpha} = \vone^{\top} \bigl(\D - (1 - \alpha) \A\bigr) = \alpha \, \vone^{\top} \D$.
For any $\x \in \R^V$, multiplying \cref{eqn:residue-def} by $\vone^{\top}$ gives
\begin{align}
	\vone^{\top} \D \x + \vone^{\top} \r(\x) = 1. \label{eqn:mass-identity}
\end{align}
Additionally, by inspecting \cref{eqn:residue-def} entrywise, we have that, if $\supp(\x) \subseteq S$ and $s \in S$, then
\begin{align}
	\r(\x)(v) = \frac{1 - \alpha}{\alpha} \sum_{u \in N(v) \cap S} \x(u), \qquad \forall \, v \in V \setminus S. \label{eqn:frontier-residue-formula}
\end{align}
In particular, $\r(\x)(v) = 0$ for every $v \notin S \cup \partial S$.

In our algorithm, we use nearly-linear-time approximate solvers for SDD linear systems as a black box.
The solvers' accuracy guarantee is given in terms of the following \textit{energy norm}: for a vector $\x$ and a symmetric positive definite matrix $\S$, define $\|\x\|_{\S} := \sqrt{\x^{\top} \S \x} = \bigl\|\S^{1/2}\x\bigr\|_2$.

\begin{theorem}[SDD solver~\cite{spielman2004nearly,jambulapati2025ultrasparse}] \label{thm:SDD_solver}
	There is a randomized algorithm that, given a positive definite SDD matrix $\S \in \R^{N \times N}$, a vector $\b \in \R^N$, an accuracy parameter $\mu \in (0,1)$, and a failure probability $\delta \in (0,1)$, with probability at least $1 - \delta$ returns a vector $\tx \in \R^N$ such that
	\begin{align}
		\bigl\| \tx - \S^{-1}\b \bigr\|_{\S} \le \mu \bigl\| \S^{-1}\b \bigr\|_{\S} \label{eqn:solver-guarantee}
	\end{align}
	in time $O\bigl( \nnz(\S) \polylog\bigl(N / (\mu\delta)\bigr) \bigr)$.
\end{theorem}

Note that the term $\polylog\bigl(N / (\mu\delta)\bigr)$ in the complexity bound above can be small (e.g., using \cite{jambulapati2025ultrasparse}), but we leave it implicit since it is not central to our results.

We now present the approximate active-set algorithm in \cref{alg:active-set}.
We write it as a unified framework that can be used for both ACL approximate PageRank and $\ell_1$-regularized PageRank.
Specifically, we parameterize the active-set algorithm by an internal residue level $\lambda$ and an activation gap $\kappa$.
The target residue on the active set is $\lambda \d$, while the algorithm expands at the threshold $(\lambda + \kappa) \d$, and we also set the parameters for the SDD solvers according to these parameters.
For the task of computing an $\eps$-approximate PageRank vector, we set $\lambda := 0.5\eps$ and $\kappa := 0.5\eps$, in which case the algorithm framework exactly matches the conceptual one in \cref{sec:overview}.

\begin{algorithm}[ht]
	\DontPrintSemicolon
	\caption{$\activeSet(s,\alpha,\lambda,\kappa,\delta)$} \label{alg:active-set}
	\KwIn{seed node $s \in V$, teleportation parameter $\alpha \in (0,1)$, internal residue level $\lambda > 0$, activation gap $\kappa > 0$, failure probability $\delta \in (0,1)$}
	\KwOut{a degree-normalized vector $\tx$ and an approximate PageRank vector $\tp = \D\tx$}
	\If{$(\lambda + \kappa) \, \d(s) \ge 1$}
	{
		\Return{$\tx \gets \vzero$ \textup{and} $\tp \gets \vzero$} \;
	}
	$S \gets \{s\}$ \;
	\While{\textup{true}}
	{
		invoke the SDD solver of \cref{thm:SDD_solver} on $\L_{\alpha,S} \, \x_S = \alpha(\e_s|_S - \lambda \, \d_S)$ with parameters $\bigl(0.01\alpha \, \min\{\lambda,\kappa\},\lambda\delta\bigr)$, obtaining an approximate solution $\tx_S$ \;
		$\tr(v) \gets \frac{1-\alpha}{\alpha}\sum_{u \in N(v)\cap S}\tx_S(u)$ for each $v \in \partial S$ \;
		$T \gets \bigl\{v \in \partial S: \tr(v) > (\lambda + \kappa) \, \d(v)\bigr\}$ \;
		\If{$T = \varnothing$}
		{
			extend $\tx_S$ to $\tx \in \R^V$ by setting $\tx_{V\setminus S} \gets \vzero$ \;
			\Return{$\tx$ \textup{and} $\tp \gets \D\tx$} \;
		}
		$S \gets S \cup T$ \;
	}
\end{algorithm}

\subsection{Correctness Analysis}

We first introduce the notations for the exact solutions to the restricted linear systems that arise in the active-set algorithm.
For an active set $S$ containing $s$, define $\x^{(S)} \in \R^V$ by
\begin{align}
	\x^{(S)}_S := \alpha \, \L_{\alpha,S}^{-1} \bigl(\e_s|_S - \lambda \, \d_S\bigr), \qquad \x^{(S)}_{V\setminus S} := \vzero. \label{eqn:exact-solution}
\end{align}
Plugging \cref{eqn:exact-solution} into \cref{eqn:residue-def} verifies that $\r\bigl(\x^{(S)}\bigr)_S = \lambda \, \d_S$.

The following lemma formalizes the monotonicity property suggested by the infinite-push interpretation in \cref{sec:overview}.

\begin{lemma} \label[lemma]{lem:active-set-expansion}
	Suppose $s \in S \subseteq V$ and $T \subseteq \partial S$ satisfies $\r\bigl(\x^{(S)}\bigr)_T > \lambda \, \d_T$.
	Then $\x^{(S \cup T)} \ge \x^{(S)}$ and
	\begin{align}
		\x^{(S \cup T)}(v) \ge \frac{\alpha}{\d(v)}\bigl(\r\bigl(\x^{(S)}\bigr)(v) - \lambda \, \d(v)\bigr), \qquad \forall \, v \in T. \label{eqn:active-set-expansion}
	\end{align}
\end{lemma}

The property \eqref{eqn:active-set-expansion} in this lemma also has an intuitive interpretation: in the infinite-push process, if the active set is expanded to include $T$, then for any node $v \in T$, at least one push operation would be performed on $v$, the first of which would increment its reserve by at least $\alpha \bigl(\r\bigl(\x^{(S)}\bigr)(v) - \lambda \, \d(v)\bigr)$.

\begin{proof}[Proof of \cref{lem:active-set-expansion}]
	Since both $\x^{(S)}$ and $\x^{(S \cup T)}$ are zero outside $S \cup T$, the relationship between $\x$ and $\r(\x)$ in \cref{eqn:residue-invariant} gives
	\begin{align}
		\L_{\alpha,S \cup T}\bigl(\x^{(S \cup T)} - \x^{(S)}\bigr)_{S \cup T} = \alpha \bigl(\r\bigl(\x^{(S)}\bigr) - \r\bigl(\x^{(S \cup T)}\bigr)\bigr)_{S \cup T}. \label{eqn:active-set-expansion-system}
	\end{align}
	We have $\r\bigl(\x^{(S)}\bigr)_S = \r\bigl(\x^{(S \cup T)}\bigr)_S = \lambda \, \d_S$.
	For $v \in T$, we have $\r\bigl(\x^{(S \cup T)}\bigr)(v) = \lambda \, \d(v)$, whereas the assumption gives $\r\bigl(\x^{(S)}\bigr)(v) > \lambda \, \d(v)$.
	Hence the right-hand side of \cref{eqn:active-set-expansion-system} is nonnegative.
	$\L_{\alpha,S \cup T}$ is symmetric positive definite with nonpositive off-diagonal entries, so it is a nonsingular $M$-matrix and $\L_{\alpha,S \cup T}^{-1}$ is entrywise nonnegative.
	So multiplying \cref{eqn:active-set-expansion-system} by $\L_{\alpha,S \cup T}^{-1} \ge \mzero$ proves $\x^{(S \cup T)} \ge \x^{(S)}$.

	For $v \in T$, using $\L_{\alpha} = \D - (1 - \alpha)\A$, the row indexed by $v$ in \eqref{eqn:active-set-expansion-system} gives
	\begin{align*}
		& \phantom{{}={}} \d(v) \bigl(\x^{(S \cup T)}(v) - \x^{(S)}(v)\bigr) - (1 - \alpha)\sum_{u \in N(v)\cap(S \cup T)}\bigl(\x^{(S \cup T)}(u) - \x^{(S)}(u)\bigr) \\
		& = \alpha \bigl(\r\bigl(\x^{(S)}\bigr)(v) - \lambda \, \d(v)\bigr).
	\end{align*}
	Since $\x^{(S)}(v) = 0$ and $\x^{(S \cup T)} \ge \x^{(S)}$ by the first claim, it follows that
	\begin{align*}
		\d(v) \, \x^{(S \cup T)}(v) \ge \alpha \bigl(\r\bigl(\x^{(S)}\bigr)(v) - \lambda \, \d(v)\bigr).
	\end{align*}
	Dividing by $\d(v)$ proves the second claim.
\end{proof}

The next lemma shows that the accuracy of the SDD solvers guarantees that $\tx_S$ and $\tr$ are close to their exact counterparts $\x^{(S)}_S$ and $\r\bigl(\x^{(S)}\bigr)$, respectively.

\begin{lemma} \label[lemma]{lem:active-set-solver-error}
	Suppose that $(\lambda + \kappa) \, \d(s) < 1$ and in an iteration of \cref{alg:active-set}, the set $S$ satisfies $\x^{(S)} \ge \vzero$.
	If the SDD solve in this iteration returns a $\tx_S$ that satisfies the accuracy guarantee \eqref{eqn:solver-guarantee} in \cref{thm:SDD_solver}, then, extending $\tx_S$ by zero outside $S$ to obtain $\tx \in \R^V$ and letting $\tr := \r(\tx)$, we have
	\begin{alignat}{2}
		\bigl|\tr(v) - \lambda \, \d(v)\bigr| & \le 0.1\min\{\lambda,\kappa\} \, \d(v), & \qquad & \forall \, v \in S, \label{eqn:active-set-active-residue-error} \\
		\bigl|\tx_S(v) - \x^{(S)}(v)\bigr| & \le 0.1 \alpha \kappa, & \qquad & \forall \, v \in S, \label{eqn:active-set-coordinate-error} \\
		\bigl|\tr(v) - \r\bigl(\x^{(S)}\bigr)(v)\bigr| & \le 0.1 \kappa \, \d(v), & \qquad & \forall \, v \in \partial S. \label{eqn:active-set-frontier-error}
	\end{alignat}
\end{lemma}

The proof of \cref{lem:active-set-solver-error} is deferred to the next subsection.
Here, we establish the properties of \cref{alg:active-set} using \cref{lem:active-set-solver-error}.

\begin{lemma} \label[lemma]{lem:active-set}
	With probability at least $1 - \delta$, \cref{alg:active-set} returns a $(\lambda + \kappa)$-approximate PageRank vector $\tp \ge \vzero$ with $\vol\bigl(\supp(\tp)\bigr) \le 1 / \lambda$, and it runs in time
	\begin{align*}
		O\left( \frac{1}{\lambda^2} \polylog\frac{1}{\alpha\lambda\kappa\delta} \right).
	\end{align*}
\end{lemma}

\begin{proof}
	If $(\lambda + \kappa) \, \d(s) \ge 1$, the algorithm directly outputs $\tp = \vzero$.
	In this case, its corresponding residue vector $\r = \e_s$ satisfies $\vzero \le \r \le (\lambda + \kappa) \d$, so $\tp$ is a $(\lambda + \kappa)$-approximate PageRank vector and the lemma holds.
	
	Now assume $(\lambda + \kappa) \, \d(s) < 1$ and every SDD solver call in the algorithm satisfies the accuracy guarantee in \cref{thm:SDD_solver}.
	We prove by induction that for each active set $S$ encountered by the algorithm, $\x^{(S)}(v) > 0.9\alpha\kappa$ for every $v \in S$.
	For the base case $S = \{s\}$, by \cref{eqn:exact-solution} we have $\x^{(S)}(s) = \alpha \bigl(1 / \d(s) - \lambda\bigr) > \alpha\kappa$, so the base case holds.

	For the inductive step, suppose that the current active set $S$ satisfies the inductive hypothesis and $T := \bigl\{ v \in \partial S: \tr(v) > (\lambda + \kappa) \, \d(v) \bigr\}$ is nonempty.
	For every $v \in T$, \eqref{eqn:active-set-frontier-error} gives
	\begin{align*}
		\r\bigl(\x^{(S)}\bigr)(v) \ge \tr(v) - 0.1 \kappa \, \d(v) > (\lambda + 0.9\kappa) \, \d(v).
	\end{align*}
	Thus, the conditions of \cref{lem:active-set-expansion} are satisfied, giving $\x^{(S \cup T)}(v) \ge \x^{(S)}(v) > 0.9\alpha\kappa$ for $v \in S$ and
	\begin{align*}
		\x^{(S \cup T)}(v) \ge \frac{\alpha}{\d(v)}\bigl(\r\bigl(\x^{(S)}\bigr)(v) - \lambda \, \d(v)\bigr) > 0.9\alpha\kappa
	\end{align*}
	for $v \in T$.
	This completes the induction.

	Next, we inspect the property of the output $\tp = \D\tx$.
	Let $S'$ be the active set at termination and $\tr := \r(\tx)$.
	We have shown that $\x^{(S')}(v) > 0.9\alpha\kappa$ for every $v \in S'$, so \eqref{eqn:active-set-active-residue-error} gives $\vzero \le \tr_{S'} \le (\lambda + \kappa) \d_{S'}$.
	Additionally, \eqref{eqn:active-set-coordinate-error} guarantees that $\tx_{S'}(v) \ge \x^{(S')}(v) - 0.1\alpha\kappa \ge 0$ for every $v \in S'$, and thus $\tr_{V \setminus S'} \ge \vzero$ by \cref{eqn:frontier-residue-formula}.
	On the other hand, the algorithm's stopping rule guarantees that $\tr_{V \setminus S'} \le (\lambda + \kappa) \d_{V \setminus S'}$.
	Together, $\vzero \le \tr \le (\lambda + \kappa) \d$, so $\tp$ is a $(\lambda + \kappa)$-approximate PageRank vector.
	On the other hand, $\x^{(S')} \ge \vzero$ and $\r\bigl(\x^{(S')}\bigr) \ge \vzero$ by \cref{eqn:frontier-residue-formula}, so the mass identity \eqref{eqn:mass-identity} gives $1 = \vone^{\top} \D\x^{(S')} + \vone^{\top} \r\bigl(\x^{(S')}\bigr) \ge \vone_{S'}^{\top} \r\bigl(\x^{(S')}\bigr)_{S'} = \lambda \vol(S')$, which leads to $\vol\bigl(\supp(\tp)\bigr) \le \vol(S') \le 1 / \lambda$, as desired.
	
	For an iteration with active set $S$, the involved SDD system has dimension at most $\vol(S') \le 1 / \lambda$ and also has $O(1 / \lambda)$ nonzero entries, and the required accuracy and failure probability for the solver is $0.01\alpha \min\{\lambda,\kappa\}$ and $\lambda\delta$, respectively.
	Additionally, it takes $O\bigl(\vol(S)\bigr)$ time for inspecting the outer boundary $\partial S$.
	Since every nonterminal iteration adds at least one node and $|S'| \le \vol(S') \le 1 / \lambda$, there are at most $1 / \lambda$ iterations.
	Summing over the iterations and using the running-time bound in \cref{thm:SDD_solver} yield the desired total running-time bound.

	Finally, we bound the probability that some invocation of the SDD solver fails to satisfy the accuracy guarantee in \cref{thm:SDD_solver}.
	Consider the probability that the first failure occurs at the $i$-th call, denoted by $p_i$.
	By the analysis above, the success of the first $1 / \lambda$ calls implies that the algorithm terminates, so $p_i$ can be nonzero only if $i \le 1 / \lambda$.
	For each $i$, $p_i$ is at most the probability that the $i$-th call fails, which is bounded by $\lambda\delta$ by \cref{thm:SDD_solver}.
	Therefore, summing $p_i$ over $1 \le i \le \lfloor 1 / \lambda \rfloor$ yields that the probability that some SDD solver call fails is at most $\delta$, completing the proof.
\end{proof}

Now the main result \cref{thm:ACL} for computing an $\eps$-approximate PageRank vector is a direct consequence of \cref{lem:active-set}.

\begin{proof}[Proof of \cref{thm:ACL}]
	Applying \cref{lem:active-set} with $\lambda := 0.5\eps$ and $\kappa := 0.5\eps$ proves the theorem.
\end{proof}

\subsection{Bounding the Error Introduced by SDD Solvers}

This subsection proves \cref{lem:active-set-solver-error}.
We will use the following standard facts.

\paragraph{Loewner order.}
For symmetric matrices $\M$ and $\N$ of the same dimension, we write $\M \preceq \N$ if $\N - \M$ is positive semidefinite and $\M \prec \N$ if $\N - \M$ is positive definite, and use $\succeq$ and $\succ$ analogously.
Taking the same principal submatrix of both matrices preserves the Loewner order; for any $c > 0$, $\mzero \prec \M \preceq c\N$ implies $\M^{-1} \succeq c^{-1} \N^{-1} \succ \mzero$.

\paragraph{Cauchy-Schwarz inequalities with energy norms.}
Consider a matrix $\M \succ \mzero$ and vectors $\x$ and $\y$ of the same dimension.
Applying the Cauchy-Schwarz inequality to $\M^{1/2}\x$ and $\M^{1/2}\y$ gives
\begin{align}
	\bigl|\x^{\top}\M\y\bigr| \le \bigl\|\M^{1/2}\x\bigr\|_2\bigl\|\M^{1/2}\y\bigr\|_2 = \|\x\|_{\M}\|\y\|_{\M}. \label{eqn:energy-Cauchy-Schwarz}
\end{align}
Similarly, applying it to $\M^{-1/2}\x$ and $\M^{1/2}\y$ gives
\begin{align}
	\bigl|\x^{\top}\y\bigr| \le \bigl\|\M^{-1/2}\x\bigr\|_2\bigl\|\M^{1/2}\y\bigr\|_2 = \|\x\|_{\M^{-1}}\|\y\|_{\M}. \label{eqn:primal-dual-Cauchy-Schwarz}
\end{align}

\paragraph{Facts for the matrix $\L_{\alpha,S}$.}
The standard Laplacian bounds give $\mzero \preceq \D - \A \preceq 2\D$.
Taking principal submatrices and using $\L_{\alpha,S} = \alpha\D_S + (1 - \alpha)(\D_S - \A_S)$ give $\alpha\D_S \preceq \L_{\alpha,S} \preceq 2\D_S$ and therefore $\L_{\alpha,S}^{-1} \preceq \alpha^{-1} \D_S^{-1}$.
We also have
\begin{align}
	\e_v^{\top} \L_{\alpha,S} \, \e_v = \d(v), \quad \e_v^{\top} \L_{\alpha,S}^{-1} \, \e_v \le \frac{1}{\alpha \, \d(v)}, \qquad \forall \, v \in S. \label{eqn:active-set-coordinate-energy-bounds}
\end{align}

\begin{proof}[Proof of \cref{lem:active-set-solver-error}]
	Recall that $\r\bigl(\x^{(S)}\bigr)_S = \lambda \, \d_S$ and $\x^{(S)} \ge \vzero$ implies that $\r\bigl(\x^{(S)}\bigr) \ge \vzero$.
	The assumptions and the mass identity \eqref{eqn:mass-identity} give $1 = \vone^{\top} \D\x^{(S)} + \vone^{\top} \r\bigl(\x^{(S)}\bigr) \ge \vone_S^{\top} \r\bigl(\x^{(S)}\bigr)_S = \lambda \vol(S)$, so $\vol(S) \le 1 / \lambda$.

	Let $\b_S := \alpha(\e_s|_S - \lambda \, \d_S)$.
	Using $\L_{\alpha,S}^{-1} \preceq \alpha^{-1} \D_S^{-1}$, we obtain
	\begin{align*}
		\|\b_S\|_{\L_{\alpha,S}^{-1}}^2 & = \b_S^{\top}\L_{\alpha,S}^{-1} \, \b_S \le \frac{1}{\alpha}\b_S^{\top}\D_S^{-1}\b_S = \alpha \biggl(\frac{\bigl( 1 - \lambda \, \d(s) \bigr)^2}{\d(s)} + \lambda^2\sum_{v \in S\setminus\{s\}}\d(v)\biggr) \\
		& \le \alpha\bigl(1 + \lambda^2\vol(S)\bigr) \le \alpha(1 + \lambda) \le 2\alpha.
	\end{align*}
	Since $\x^{(S)}_S = \L_{\alpha,S}^{-1} \, \b_S$, we have $\bigl\|\x^{(S)}_S\bigr\|_{\L_{\alpha,S}}^2 = \|\b_S\|_{\L_{\alpha,S}^{-1}}^2 \le 2\alpha$.
	Let $\vdelta_{\x} := \tx_S - \x^{(S)}_S \in \R^S$.
	The solver's accuracy guarantee in \cref{thm:SDD_solver} with $\mu := 0.01 \alpha \, \min\{\lambda,\kappa\}$ therefore gives
	\begin{align}
		\bigl\|\vdelta_{\x}\bigr\|_{\L_{\alpha,S}} \le 0.01 \alpha \, \min\{\lambda,\kappa\} \, \sqrt{2\alpha} \le 0.02 \alpha^{3/2} \min\{\lambda,\kappa\}. \label{eqn:active-set-energy-error}
	\end{align}

	For $v \in S$, \cref{eqn:residue-def} gives
	\begin{align*}
		\tr(v) - \lambda \, \d(v) = -\frac{1}{\alpha} \e_v^{\top} \L_{\alpha,S} \, \vdelta_{\x},
	\end{align*}
	which together with \eqref{eqn:energy-Cauchy-Schwarz}, \eqref{eqn:active-set-coordinate-energy-bounds}, and \eqref{eqn:active-set-energy-error} gives
	\begin{align*}
		\bigl|\tr(v) - \lambda \, \d(v)\bigr| \le \frac{1}{\alpha} \|\e_v\|_{\L_{\alpha,S}}\bigl\|\vdelta_{\x}\bigr\|_{\L_{\alpha,S}} \le \frac{1}{\alpha} \sqrt{\d(v)} \cdot 0.02 \alpha^{3/2} \min\{\lambda,\kappa\} \le 0.1 \min\{\lambda,\kappa\} \, \d(v),
	\end{align*}
	proving \eqref{eqn:active-set-active-residue-error}.
	On the other hand, applying \eqref{eqn:primal-dual-Cauchy-Schwarz}, \eqref{eqn:active-set-coordinate-energy-bounds}, and \eqref{eqn:active-set-energy-error} gives
	\begin{align*}
		\bigl|\tx_S(v) - \x^{(S)}(v)\bigr| \le \|\e_v\|_{\L_{\alpha,S}^{-1}} \bigl\|\vdelta_{\x}\bigr\|_{\L_{\alpha,S}} \le \frac{1}{\sqrt{\alpha \, \d(v)}} \cdot 0.02 \alpha^{3/2} \min\{\lambda,\kappa\} \le 0.1 \alpha \kappa,
	\end{align*}
	proving \eqref{eqn:active-set-coordinate-error}.

	For $v \in \partial S$, \cref{eqn:frontier-residue-formula} gives
	\begin{align*}
		\tr(v) - \r\bigl(\x^{(S)}\bigr)(v) = \frac{1 - \alpha}{\alpha}\sum_{u \in N(v) \cap S}\vdelta_{\x}(u).
	\end{align*}
	To bound the absolute sum, we write each summand as $\vdelta_{\x}(u) = \sqrt{\d(u)} \, \vdelta_{\x}(u) / \sqrt{\d(u)}$ and apply the Cauchy-Schwarz inequality, giving
	\begin{align*}
		\Biggl|\sum_{u \in N(v) \cap S}\vdelta_{\x}(u)\Biggr| & = \Biggl|\sum_{u \in N(v) \cap S}\bigl(\sqrt{\d(u)} \, \vdelta_{\x}(u)\bigr) \, \frac{1}{\sqrt{\d(u)}}\Biggr| \\
		& \le \biggl( \sum_{u \in N(v) \cap S} \d(u) \, \vdelta_{\x}(u)^2 \biggr)^{1/2} \biggl( \sum_{u \in N(v) \cap S} \frac{1}{\d(u)} \biggr)^{1/2} \\
		& \le \|\vdelta_{\x}\|_{\D_S} \sqrt{\bigl|N(v) \cap S\bigr|} \le \|\vdelta_{\x}\|_{\D_S}\sqrt{\d(v)},
	\end{align*}
	where we used $\sum_{u \in N(v) \cap S} \d(u) \, \vdelta_{\x}(u)^2 \le \sum_{u \in S} \d(u) \, \vdelta_{\x}(u)^2 = \|\vdelta_{\x}\|_{\D_S}^2$ and $\sum_{u \in N(v) \cap S} \frac{1}{\d(u)} \le \bigl|N(v) \cap S\bigr| \le \d(v)$ in the last line.
	On the other hand, since $\L_{\alpha,S} \succeq \alpha\D_S$, we have $\|\vdelta_{\x}\|_{\D_S} \le \frac{1}{\sqrt{\alpha}} \|\vdelta_{\x}\|_{\L_{\alpha,S}}$.
	Combining, we obtain
	\begin{align*}
		& \phantom{{}={}} \bigl|\tr(v) - \r\bigl(\x^{(S)}\bigr)(v)\bigr| \le \frac{1 - \alpha}{\alpha} \|\vdelta_{\x}\|_{\D_S} \sqrt{\d(v)} \le \frac{1 - \alpha}{\alpha^{3/2}} \|\vdelta_{\x}\|_{\L_{\alpha,S}} \sqrt{\d(v)} \\
		& \le \frac{1}{\alpha^{3/2}} \, 0.02 \alpha^{3/2} \min\{\lambda,\kappa\} \sqrt{\d(v)} \le 0.1 \kappa \, \d(v).
	\end{align*}
	This proves \eqref{eqn:active-set-frontier-error} and completes the proof.
\end{proof}

\section{Solving $\ell_1$-Regularized PageRank} \label{sec:regularized}

This section proves \cref{thm:regularized}.
Recall that the $\ell_1$-regularized PageRank problem is to compute an additive $\xi$-minimizer of $\psi(\cdot)$ defined in \cref{eqn:regularized-objective}.
Throughout, we let $\xast$ be the unique minimizer of $\psi(\cdot)$ and $\Sast := \supp(\xast)$, set $\tau := \min\bigl\{\rho,\xi / (2\alpha)\bigr\}$, and consider invoking \cref{alg:active-set} with internal residue level $\lambda := \rho$ and activation gap $\kappa := \tau$.

Based on applying \cref{lem:active-set} with these parameter settings, we only need to additionally show that no node outside $\Sast$ is activated throughout the algorithm and to bound the objective error.
We will use the following facts that are adapted and derived from the variational characterization of PageRank given in \cite{fountoulakis2019variational}.

\begin{lemma}[\cite{fountoulakis2019variational}] \label[lemma]{lem:regularized-facts}
	It holds that $\psi(\cdot)$ has a unique minimizer $\xast \ge \vzero$, $\vol(\Sast) \le 1 / \rho$, and
	\begin{align*}
		\begin{cases}
			\r(\xast)(v) = \rho \, \d(v), & \text{if } v \in \Sast, \\
			0 \le \r(\xast)(v) \le \rho \, \d(v), & \text{if } v \notin \Sast.
		\end{cases}
	\end{align*}
	Furthermore, if $\rho \, \d(s) < 1$, then $s \in \Sast$; if $\rho \, \d(s) \ge 1$, then $\xast = \vzero$.
\end{lemma}

For $\x \ge \vzero$, we have $\|\D\x\|_1 = \d^{\top}\x$.
We use
$\nabla\psi(\x) := \L_{\alpha} \, \x - \alpha \, \e_s + \alpha\rho \d = \alpha \bigl(\rho \d - \r(\x)\bigr)$
to denote a valid subgradient of $\psi(\cdot)$ at $\x$, where we used $\r(\x) := \e_s - \alpha^{-1} \L_{\alpha} \, \x$ as in \cref{eqn:residue-def}.
Consequently, convexity gives
\begin{align}
	\psi(\y) - \psi(\x) \ge \bigl\langle \nabla\psi(\x),\y-\x \bigr\rangle, \quad \forall \, \x,\y \ge \vzero. \label{eqn:regularized-subgradient-inequality}
\end{align}
The only optimization facts that we will use are \cref{lem:regularized-facts} and the convexity property~\eqref{eqn:regularized-subgradient-inequality}.

\begin{lemma} \label[lemma]{lem:no-false-activation}
	Assume $s \in S \subseteq \Sast$.
	We have $\x^{(S)}_S \le \xast_S$.
	Furthermore, if $v \in \partial S$ satisfies $\r\bigl(\x^{(S)}\bigr)(v) > \rho \, \d(v)$, then $v \in \Sast$.
\end{lemma}

\begin{proof}
	The optimality conditions in \cref{lem:regularized-facts} and \cref{eqn:residue-def} imply that $(\L_{\alpha} \, \xast)_{\Sast} = \alpha (\e_s|_{\Sast} - \rho \, \d_{\Sast})$.
	Restricting this equation to the rows indexed by $S$ gives
	\begin{align*}
		\L_{\alpha,S} \, \xast_S + \L_{\alpha}[S,\Sast \setminus S] \, \xast_{\Sast \setminus S} = \alpha (\e_s|_S - \rho \, \d_S),
	\end{align*}
	where $\L_{\alpha}[S,\Sast \setminus S]$ denotes the submatrix of $\L_{\alpha}$ with rows indexed by $S$ and columns indexed by $\Sast \setminus S$.
	Since the entries in $\L_{\alpha}[S,\Sast \setminus S]$ are nonpositive and $\xast_{\Sast \setminus S} \ge \vzero$ by \cref{lem:regularized-facts}, we have
	\begin{align*}
		\L_{\alpha,S} \, \xast_S \ge \alpha (\e_s|_S - \rho \, \d_S) = \L_{\alpha,S} \, \x^{(S)}_S.
	\end{align*}
	Since $\L_{\alpha,S}^{-1}$ is entrywise nonnegative, multiplying this inequality by $\L_{\alpha,S}^{-1}$ proves $\x^{(S)}_S \le \xast_S$.

	For the second claim, suppose for contradiction that such a $v \notin \Sast$.
	Since $S \subseteq \Sast$, the residue formula in \cref{eqn:frontier-residue-formula} and the first claim yield
	\begin{align*}
		\r\bigl(\x^{(S)}\bigr)(v) & = \frac{1 - \alpha}{\alpha}\sum_{u \in N(v) \cap S}\x^{(S)}(u) \le \frac{1 - \alpha}{\alpha}\sum_{u \in N(v) \cap S} \xast(u) \\
		& \le \frac{1 - \alpha}{\alpha} \sum_{u \in N(v) \cap \Sast} \xast(u) = \r(\xast)(v).
	\end{align*}
	The optimality conditions at $v \notin \Sast$ imply $\r(\xast)(v) \le \rho \, \d(v)$.
	Hence $\r\bigl(\x^{(S)}\bigr)(v) \le \rho \, \d(v)$, contradicting the assumption that $\r\bigl(\x^{(S)}\bigr)(v) > \rho \, \d(v)$.
	Therefore, $v \in \Sast$, as desired.
\end{proof}

\begin{proof}[Proof of \cref{thm:regularized}]
	Assume that every SDD solver call made by the algorithm satisfies the accuracy guarantee in \cref{thm:SDD_solver}, which happens with probability at least $1-\delta$ by the same analysis in the proof of \cref{lem:active-set}.
	By \cref{lem:active-set}, the returned $\tx$ is nonnegative and $\tr := \r(\tx)$ satisfies $\vzero \le \tr \le (\rho+\tau) \d$.
	Thus, $\D\tx$ is a $(\rho+\tau)$-approximate PageRank vector, which is also a $(2\rho)$-approximate PageRank vector since $\tau \le \rho$.
	We only need to establish the objective accuracy and the running-time bound involving $\Sast$.
	
	By \cref{lem:regularized-facts}, $\xast \ge \vzero$ and $\r(\xast) \ge \vzero$, so applying the mass identity in \cref{eqn:mass-identity} gives $\vone^{\top} \D \xast = \d^{\top} \xast \le 1$.

	If $\tx = \vzero$, applying \eqref{eqn:regularized-subgradient-inequality} with $\x = \vzero$ and $\y = \xast$, together with $\tr \le (\rho+\tau) \d$, gives
	\begin{align*}
		\psi(\vzero) - \psi(\xast) \le -\bigl\langle \nabla\psi(\vzero),\xast \bigr\rangle = -\alpha (\rho \d - \tr)^{\top} \xast \le \alpha\tau \, \d^{\top}\xast \le \alpha\tau \le \xi.
	\end{align*}
	Thus, the theorem holds when the algorithm returns $\vzero$.

	Assume from now on that $\tx \ne \vzero$.
	We prove by induction that every active set $S$ encountered by the algorithm satisfies $S \subseteq \Sast$.
	The base case follows from \cref{lem:regularized-facts} and $(\rho+\tau) \, \d(s) < 1$.
	For the inductive step, suppose $S \subseteq \Sast$ and the algorithm adds a nonempty set $T$ to $S$.
	By the same analysis as in the proof of \cref{lem:active-set}, it holds that $\r\bigl(\x^{(S)}\bigr)(v) > (\rho + 0.9\tau) \, \d(v) > \rho \, \d(v)$ for every $v \in T$.
	Therefore, \cref{lem:no-false-activation} gives $T \subseteq \Sast$, proving the induction.

	Next, we bound the objective error from the residue guarantees.
	Let $S'$ be the final active set.
	For every $v \in S'$, applying \eqref{eqn:active-set-active-residue-error} with $\lambda := \rho$ and $\kappa := \tau$ gives
	\begin{align*}
		\bigl|\tr(v) - \rho \, \d(v)\bigr| \le 0.1\min\{\rho,\tau\} \, \d(v) = 0.1\tau \, \d(v).
	\end{align*}
	Since $\nabla\psi(\tx) = \alpha(\rho \d - \tr)$, it follows that
	\begin{align*}
		\nabla\psi(\tx)(v) = \alpha\bigl(\rho \, \d(v) - \tr(v)\bigr) \le \alpha \, \bigl|\rho \, \d(v) - \tr(v)\bigr| \le 0.1\alpha\tau \, \d(v).
	\end{align*}
	On the other hand, the residue bound $\tr \le (\rho+\tau) \d$ gives
	\begin{align*}
		\nabla\psi(\tx)(v) = \alpha\bigl(\rho \, \d(v) - \tr(v)\bigr) \ge -\alpha\tau \, \d(v), \quad \forall \, v \in V.
	\end{align*}
	Since $\tx \ge \vzero$ and $\tr \ge \vzero$, the mass identity also gives $\d^{\top}\tx \le 1$.
	Applying \eqref{eqn:regularized-subgradient-inequality} with $\x = \tx$ and $\y = \xast$, and using $\tx_{V\setminus S'} = \vzero$, we obtain
	\begin{align*}
		\psi(\tx) - \psi(\xast) & \le \bigl\langle \nabla\psi(\tx),\tx-\xast \bigr\rangle = \sum_{v \in S'}\nabla\psi(\tx)(v)\,\tx(v) - \sum_{v \in V}\nabla\psi(\tx)(v)\,\xast(v) \\
		& \le 0.1\alpha\tau \, \d^{\top}\tx + \alpha\tau \, \d^{\top}\xast \le 1.1\alpha\tau \le \xi,
	\end{align*}
	as desired.

	Finally, $S' \subseteq \Sast$ guarantees that there are at most $|\Sast|$ iterations and each involved SDD system has dimension at most $|\Sast| \le 1 / \rho$ and has at most $\tvol(\Sast)$ nonzero entries.
	Additionally, it takes $O\bigl(\vol(\Sast)\bigr)$ time for inspecting the outer boundary $\partial S$.
	These bounds yield the desired running-time bound and finish the proof.
\end{proof}

\section{Concluding Remarks}

Our algorithm computes an ACL $\eps$-approximate PageRank vector in $\tO(1 / \eps^2)$ time, whereas the classic \push algorithm has running time $O\bigl(1 / (\alpha\eps)\bigr)$.
Neither bound uniformly dominates the other, so the currently known upper bound can be written as
\begin{align*}
	\tO\Bigl(\min\Bigl\{\frac{1}{\eps^2},\frac{1}{\alpha\eps}\Bigr\}\Bigr).
\end{align*}
As a result, our upper bound rules out general lower bounds of the form $\Omega\bigl(1 / (\alpha\eps)\bigr)$ or $\Omega\bigl(1 / (\sqrt{\alpha} \, \eps)\bigr)$ that hold uniformly over all reasonable parameter choices in this model.
Meaningful lower bounds for this problem must therefore distinguish different parameter regimes.

An important open direction is to reuse the computation between successive SDD solves in our active-set algorithm.
These active sets are nested and the involved SDD systems are structured, but the present algorithm solves every SDD system from scratch.
An incremental or warm-started implementation of the SDD solvers could plausibly reduce the extra factor corresponding to the number of iterations and reduce the running time toward $\tO(1 / \eps)$.
Such a bound would also substantially strengthen the resulting local graph clustering algorithm, giving nearly linear dependence on the target volume while retaining only a polylogarithmic dependence on the inverse target conductance.
This perspective also suggests that, for $\ell_1$-regularized PageRank, the $1 / \sqrt{\alpha}$ dependence targeted by the open problem of Fountoulakis and Yang~\cite{fountoulakis2022open} need not be the ultimate limit.

\section{Acknowledgments}

This research was supported in part by National Natural Science Foundation of China (No. U2241212).

\printbibliography

\end{document}